\documentclass[10pt,a4paper]{article}
\pdfoutput=1
\usepackage{a4wide}
\usepackage{setspace}
\usepackage{theorem}
\usepackage{latexsym}
\usepackage{amsmath}
\usepackage{amssymb}
\usepackage{verbatim}
\usepackage{subfig}
\usepackage{graphicx}
\usepackage{url}
\usepackage{wrapfig}

\newtheorem{theorem}{Theorem}
\newtheorem{lemma}{Lemma}
\newtheorem{proposition}{Proposition}

\newtheorem{corollary}{Corollary}

\newtheorem{definition}{Definition}

\newcommand{\qed}{\hfill $\Box$}
\renewcommand{\vec}[1]{\mathbf{#1}}

\allowdisplaybreaks

\date{}
\title{
On the Inefficiency of\\
the Uniform Price Auction
\protect\footnote{
Work supported by the research project ``DDCOD'' (PE6-213). The research project is implemented within the framework of the Action ``Supporting Postdoctoral Researchers'' of the Operational Program ``Education and Lifelong Learning'' (Action’s Beneficiary: General Secretariat for Research and Technology), and is co-financed by the European Social Fund (ESF) and the Greek State.
}
}
\author{
Evangelos Markakis\quad\quad Orestis Telelis\medskip\\
Department of Informatics,\\
Athens University of Economics and Business, Greece.\\
\texttt{\{markakis,telelis\}@gmail.com}
}

\begin{document}
\bibliographystyle{plain}
\begin{titlepage}
\maketitle

\begin{abstract}
We present our results on Uniform Price Auctions, one of the standard sealed-bid multi-unit auction formats, for selling multiple identical units of a single good to multi-demand bidders. Contrary to the truthful and economically efficient multi-unit Vickrey auction, the Uniform Price Auction encourages strategic bidding and is socially inefficient in general, partly because of a "Demand Reduction" effect; bidders tend to bid for fewer (identical) units,  so as to receive them at a lower uniform price. Despite its inefficiency, the uniform pricing rule is widely popular by its appeal to the natural anticipation, that identical items should be identically priced. Application domains of its variants include sales of U.S. Treasury notes to investors, trade exchanges over the internet facilitated by popular online brokers, allocation of radio spectrum licenses etc. 

In this work we study equilibria of the Uniform Price Auction for bidders with (symmetric) submodular valuation functions, over the number of units that they win. We investigate pure Nash equilibria of the auction in undominated strategies; we produce a characterization of these equilibria that allows us to prove that a fraction $1-e^{-1}$ of the optimum social welfare is always recovered in undominated pure Nash equilibrium -- and this bound is essentially tight. Subsequently, we study the auction under the incomplete information setting and prove a bound of $4-\frac{2}{k}$ on the economic inefficiency of (mixed) Bayes Nash equilibria that are supported by undominated strategies.
\end{abstract}
\end{titlepage}

\section{Introduction}
\label{section:introduction}

We study the {\em Uniform Price Auction}, a standard multi-unit auction
format, for allocating multiple units of a single good to multi-demand bidders
within a single auction process. Multi-unit auctions are being applied in a
variety of diverse trade exchanges, including online sales over the Internet
held by various brokers~\cite{Ockenfels06}, allocation of radio spectrum
licenses~\cite{Milgrom04}, sales of U.S. Treasury notes to
investors~\cite{treasury}, and allocation of advertisement slots on Internet
sites~\cite{Edelman07}. The particular feature of the Uniform Price Auction is
a single price for every unit allocated to any bidder; this makes it a proper
representative of a wider category of uniform pricing auctions, as opposed to
{\em discriminatory pricing} ones, that sell identical units of a
single item at different prices~\cite{Ockenfels06,Krishna02}). As observed by
Milgrom in~\cite{Milgrom04}, resurgence of interest in auction design is owed
to a large extent to the success of multi-unit and -- particularly -- uniform
price auction formats. Charging a uniform price for identical items, apart
from appealing to the intuitive anticipation that identical items should be
identically priced, it eases the worries of proxy agents that bid on behalf of
their employers; they do not have to explain why they did not achieve a better
price than their competitors.

The design of {\em mechanisms} for auctioning multiple units of a single good to multi-demand bidders dates back to the seminal work of Vickrey~\cite{Vickrey61}. Since then, three sealed-bid {\em standard} multi-unit auction formats have been identi\-fied in Auction Theory~\cite{Krishna02} [Chapter 12]: the Multi-Unit Vickrey Auction, the Uniform Price Auction, and the Discriminatory Price Auction. A signifi\-cant volume of research in economics has been dedicated to identifying the properties of these standard formats~\cite{Noussair95,Engelbrecht-Wiggans98,Ausubel02,Reny99,Bresky08}. All three auctions have the same (sealed) bidding format and allocation rule, and have been studied mostly for bidders with {\em ``downward sloping''} ({\em symmetric submodular}~\cite{Lehmann06}) valuations; these prescribe that the {\em marginal} value that a bidder has for each additional unit is non-increasing. Therefore each bidder is asked to issue such a non-increasing sequence of {\em marginal} bids for the $k$ available units. The $k$ highest marginal bids win the auction and each winning bid grants its issuing bidder a distinct unit. The Multi-Unit Vickrey Auction charges according to an instantiation of the Clarke payment rule~\cite{Clarke} and is a generalization of Vickrey's celebrated single-item Second-Price mechanism. The Discriminatory Price Auction charges the winning bids as payments and it is a generalization of the First-Price Auction. The Uniform Price Auction, proposed by Friedman \cite{Friedman60}, charges per allocated unit the highest rejected (losing) marginal bid. Among these three formats, the Multi-Unit Vickrey Auction for submodular bidders retains the characteristics of the single-item Second-Price Mechanism, i.e., optimizes the Social Welfare and is truthful (it is a --weakly -- dominant strategy for every bidder to report his marginal values truthfully). Neither the Discriminatory nor the Uniform Price auctions are truthful; they encourage strategic bidding.

One of the downsides of the Uniform Price Auction is the effect of {\em Demand Reduction}, observed in~\cite{Noussair95,Engelbrecht-Wiggans98} and formalized in a general model for multi-unit auctions by Ausubel and Cramton~\cite{Ausubel02}. Bidders may have an incentive to shade their marginal bids for some units, only to win fewer ones in a lower uniform price. This effect leads to diminished revenue and inefficient allocations in equilibrium. In particular, it is known that the socially optimal allocation cannot be generally implemented in an equilibrium in (weakly) {\em undominated strategies}. Despite this effect, the variants of the Uniform Price Auction have seen extensive applications, contrary to the Vickrey auction, which has been largely overlooked in practice; implementations of variants of the standard format are offered by several online brokers~\footnote{Among them, {\tt eBay} ceased its own variant in 2009.}~\cite{Ockenfels06,Kittsteiner07} and are also being used for sales of U.S. Treasury notes to investors since 1992~\cite{treasury}. We also note that despite the Demand Reduction effect, the Uniform Price Auction does retain some interesting characteristics: overbidding any marginal value is a weakly dominated strategy, and so is any misreport of the marginal bid for the {\em first} unit.

In this work we give a detailed account of the properties of undominated pure Nash equilibria for the Uniform Price Auction, when bidders have submodular valuation functions. Subsequently, we study the economic inefficiency of pure Nash and mixed Bayes-Nash equilibria, incurred by the effect of demand reduction in the former case, and by demand reduction and incomplete information in the latter.

\medskip

\noindent {\bf Contribution.} 
We study pure Nash and (mixed) Bayes-Nash equilibria of the Uniform Price Auction. We focus first on bidders with submodular valuation functions and in Section \ref{section:undominated-pne} we give a detailed description of (pure) undominated strategies in the standard model of Uniform Price Auctions. Although these properties are mentioned or partially derived in previous works on Uniform Price Auctions, our analysis aims at clarifying some ambiguity between assumptions and implications. Additionally, we give a characterization of a subset of pure Nash equilibria in undominated strategies, that facilitates our analysis of economic inefficiency later on. 

In Section~\ref{section:pne-poa}, we study the social inefficiency of pure Nash equilibria (PNE) for submodular bidders, in undominated strategies, i.e., the Price of Anarchy (PoA) over the subset of such equilibria. We derive an upper bound of $\frac{e}{e-1}$, which states that the auction recovers in (undominated) PNE at least a fraction $1-e^{-1}$ of the welfare of the socially optimal allocation. We note here that the auction does have a socially optimal equilibrium (discussed in Section~\ref{section:model-definitions}, but not in undominated strategies; undominated PNE are known to be socially inefficient in general. As noted earlier, this is largely due to the effect of {\em Demand Reduction}~\cite{Ausubel02}, whereby a bidder shades his bids for additional units, so as to pay a lower price for the units he wins. Our analysis can thus be viewed as a quantification of this effect. For any number of units $k\geq 9$, we provide an almost matching lower bound, equal to $\left(1-e^{-1}+\frac{2}{k}\right)^{-1}$. We also discuss the social inefficiency of the auction for $k<9$ units. 

In Section~\ref{section:bne-poa}, we consider (mixed) Bayes-Nash equilibria in the {\em incomplete information} model of Harsanyi. For mixed Bayes-Nash equilibria that emerge from randomized bidding strategy profiles containing only undominated pure strategies in their support, we prove an upper bound of $4-\frac{2}{k}$ on the Price of Anarchy\footnote{This improves over a bound of $O(\log k)$ from~\cite{Markakis12} and over a bound of $\frac{4e}{e-1}$, shown subsequently in~\cite{Syrgkanis13}.}.

\section{Related Work}
\label{section:related-work}

\noindent{\bf Mult-Unit Auctions.} The Uniform Price Auction has received significant attention within the economics community. Noussair~\cite{Noussair95} and Engelbrecht-Wiggans and Kahn~\cite{Engelbrecht-Wiggans98} gave characterizations of pure Bayes-Nash equilibria under the model of independent private values of bidders, drawn from continuous distributions. They also made some initial observations on the effect of demand reduction. Ausubel and Cramton formalized demand reduction for a more general model of multi-unit auctions in~\cite{Ausubel02}, that allows also interdependent private values. Bresky showed in~\cite{Bresky08} existence of pure Bayes-Nash equilibria in the independent private values model (with continuous valuation distributions) for a large class of multi-unit auctions, including all three standard formats.

\smallskip

\noindent{\bf Simultaneous Auctions.} There has been a growing recent interest in the computer science community in analyzing auction schemes that -- although not necessarily truthful -- have an appealing simplicity and appear to achieve increased economic efficiency in equilibrium, compared to what is achievable with truthful mechanisms~\cite{Christodoulou08,Bhawalkar11,Hassidim11,Fu12,Syrgkanis12a}. Our work presents conceptual and technical resemblance to these works. 

Christodoulou, Kov\'acs and Schapira initialized in~\cite{Christodoulou08} the study of {\em Simultaneous Auctions} for bidders with combinatorial demands, where they proposed that, a set of distinct goods is sold by having an independent single-good Second-Price auction for each of them in parallel. For bidders with {\em fractionally subadditive} valuation functions (see~\cite{Feige09} for a definition), they prove a tight bound of $2$ for the Price of Anarchy, even for (mixed) Bayes-Nash equilibria. They show how pure Nash equilibria can be computed efficiently for submodular valuation functions and describe an iterative best response algorithm that converges to a pure Nash equilibrium for fractionally subadditive ones. Bhawalkar and Roughgarden proved in~\cite{Bhawalkar11} an upper bound of $2$ on the Price of Anarchy of pure Nash equilibria, for bidders with subadditive valuation functions. Feldman {\em et al.} proved for this same setting an upper bound of $4$ on the Price of Anarchy of mixed Bayes-Nash equilibria (thus, improving on a previous bound of $O(\log m)$ from~\cite{Bhawalkar11}). Finally, Fu, Kleinberg and Lavi~\cite{Fu12} showed recently a tight 
upper bound of $2$ on the Price of Anarchy of pure Nash equilibria (when they exist), for arbitrary valuation functions.

Hassidim {\em et al.}~\cite{Hassidim11} considered simultaneous First-Price auctions; first, they showed that pure Nash equilibria are efficient, when they exist, for arbitrary valuation functions. For fractionally subadditive valuation functions they proved upper bounds of $2$ and $4$ for the pure and mixed Bayes-Nash Price of Anarchy respectively. For subadditive and arbitrary valuation functions their corresponding upper bounds are $O(\log m)$ and $O(m)$. Feldman {\em et al.}~\cite{Feldman13} and Syrgkanis~\cite{Syrgkanis12c} improved the upper bound of the mixed Bayes-Nash Price of Anarchy to respectively: $2$, for subadditive valuation functions and $\frac{e}{e-1}$, for fractionally subadditive ones.

\smallskip

\noindent {\bf Sequential and Greedy Auctions.} Very recently, Syrgkanis and Tardos studied in~\cite{Syrgkanis12a} sequential First- and Second-Price auctions, motivated by the practical issue that supply may not be readily available at once. They showed that sequential First-Price auctions are efficient in subgame-perfect equilibrium. In~\cite{Syrgkanis12b} they extended their results in the incomplete information setting. Lucier and Borodin~\cite{Lucier10} analyzed the social inefficiency at (mixed) Bayes-Nash equilibrium, of combinatorial auctions for multiple distinct goods, incorporating {\em Greedy} allocation algorithms (and using appropriate adaptations of ``first'' and ``second'' pricing rules). They showed that these auctions have Price of Anarchy fairly comparable to the approximation factors of the greedy allocation algorithms, for the underlying welfare optimization problem.

\smallskip

\noindent{\bf Truthful Mechanism Design.} From the mechanism design perspective, Vickrey designed in~\cite{Vickrey61} the first truthful mechanism for auctioning multiple units ``in one go'', so as to maximize the social welfare. The Vickrey multi-unit auction is computationally efficient for a bounded number of units. Since then, computationally efficient truthful approximation mechanisms for multi-demand bidders were given by Mu'alem and Nisan in~\cite{Mualem08} and by Dobzinski and Nisan in~\cite{Dobzinski10}. These works considered several different classes of valuation functions, much more general than symmetric submodular ones. Very recently, V\"ocking gave a randomized universally truthful polynomial-time approximation scheme for bidders with general valuations~\cite{Voecking12} (a universally truthful mechanism is a probability distribution over deterministic truthful mechanisms), thus almost closing the problem. It is worth noting that in these works, the bids are accessed by the allocation algorithms through polynomially bounded many {\em value queries} to the bidders, for specific bundles of items (with the exception being the case of $k$-minded bidders, who have non-zero value for at most $k$ sizes of bundles).

\section{Definitions and Preliminaries}
\label{section:model-definitions}

We consider auctioning $k$ units of a single good to a set of $n$ bidders, denoted by $[n]$. Every bidder $i\in[n]$ has a {\em private} valuation function, defined over the quantity of units he receives, i.e., $v_i:(\{0\}\cup [k])\mapsto\Re^+$, $i=1,\dots,n$, where $v_i(0) = 0$ and each $v_i$ is {\em non-decreasing}. In this work we consider submodular 
valuation functions:
 
\begin{definition}
A valuation function $f:(\{0\}\cup [k])\mapsto\Re^+$ is called 
{\bf submodular}, if for every $x<y$, $f(x)-f(x-1)\geq f(y)-f(y-1)$.
\end{definition}

\begin{proposition}
Given $x, y\in [k]$ with $x\leq y$, any submodular valuation function $f$ satisfies $f(x)/x \geq f(y)/y$.
\label{proposition:valuation-functions}
\end{proposition}

Any non-decreasing valuation function $v:(\{0\}\cup [k])\mapsto\mathbb{R}^+$ with $v(0)=0$ can be specified also as a vector of {\em marginal values} $(m(1),m(2),\dots,m(k))$, where $m(j) = v(j) - v(j-1)$. If $v$ is submodular, then $m(1)\geq m(2)\geq\cdots\geq m(k)$. We write $m_i(\cdot)$ for the marginal value function of bidder $i$.

\paragraph{The Uniform Price Auction}
The standard Uniform Price Auction format requires that each bidder $i$ declares his whole valuation curve, by submitting a vector $\vec{b}_i$ of {\em marginal bids}, $\vec{b}_i=(b_i(1), b_i(2),\dots,b_i(k))$, satisfying $b_i(1)\geq b_i(2)\geq\cdots\geq b_i(k)$. 
Here $b_i(j)$ is the declared marginal value of $i$, for obtaining the $j$-th unit of the good. Given a bidding configuration $\vec{b}=(\vec{b}_1,\dots,\vec{b}_n)$, the allocation algorithm produces an allocation $\vec{x}(\vec{b})=(x_1(\vec{b}), x_2(\vec{b}),\dots,x_n(\vec{b}))$, as follows: each of the $k$ highest issued marginal bids grants a unit to its issuing bidder. After the allocation is completed, every bidder $i$ pays a uniform price $p(\vec{b})$ per received unit, which equals the highest rejected (marginal) bid. That is, if under bidding configuration $\vec{b}$, bidder $i$ is allocated $x_i(\vec{b})$ units in total and the uniform price is $p(\vec{b})$, $i$ pays a total of $x_i(\vec{b})\cdot p(\vec{b})$. The utility that $i$ derives under $\vec{b}$ is $u_i(\vec{b})=v_i(x_i(\vec{b}))-x_i(\vec{b})\cdot p(\vec{b})$. Given a bidding profile $\vec{b}$, we will denote by $\beta_j(\vec{b})$ the $j$-th lowest winning (maginal) bid, for $j=1,\dots,k$, i.e., the lowest winning bid is $\beta_1(\vec{b})$ and the highest one is $\beta_k(\vec{b})$.

For submodular bidders, the Uniform Price Auction admits an efficient pure Nash equilibrium; let $\vec{x}^* = (x_1^*, ...,x_n^*)$ be an optimal allocation\footnote{It is known that for submodular valuation functions on identical units, the allocation algorithm of the Uniform Price Auction produces an optimal allocation when bidders bid truthfully. This property does not hold in the case of non-identical items, where only a $2$-approximation is achieved~\cite{Lehmann06}.} of units to the bidders. Consider the profile where any winner $i$ of at least one unit in $\vec{x}^*$ bids $\vec{b}_i = (m_i(1),...,m_i(x_i^*), 0,...,0)$ and any loser bids the zero vector. It is straightforward to verify that this a Nash equilibrium. However, the strategies of losing bidders in this profile are weakly dominated, as we shall see.

Pure Nash equilibria in undominated strategies are known to suffer from a {\em demand reduction} effect~\cite{Ausubel02}; bidders may have an incentive to understate their marginal increase for the $j$-th unit onwards, for some $j>1$. This implies that equilibria in undominated strategies are generally inefficient. We show that, despite this effect, the Uniform Price Auction does quite well in approximating the optimal Social Welfare.

\medskip

\noindent {\bf Incomplete Information.} In Section~\ref{section:bne-poa} we will move to an incomplete information setting, where each bidder faces uncertainty over the other bidders' valuation functions. In particular, we will assume that every bidder $i\in[n]$ obtains his type/valuation function from a finite set $V_i$ of valuation functions, through a discrete probability distribution $\pi_i: V_i\mapsto [0,1]$, independently of the rest of the biddders; for any particular $v\in V_i$ we write $v\sim\pi_i$ to signify that it is drawn randomly from distribution $\pi_i$. The valuation function of every bidder is {\em private}. A valuation profile $\vec{v}=(v_1,\dots,v_n)\in{\cal V}=\times_i V_i$ is drawn from a {\em publicly known distribution} $\pi=\times_i\pi_i$, $\pi:{\cal V}\mapsto [0,1]$. We write accordingly $\vec{v}\sim\pi$.

Every bidder $i$ knows his own valuation function $v_i$ -- drawn from $V_i$ according to $\pi_i$, but does not know the valuation function $v_{i'}$ drawn by any other bidder $i'\neq i$. Bidder $i$ may only use his knowledge of $\pi$ to estimate $\vec{v}_{-i}$. Given the publicly known distribution $\pi$, the (possibly mixed) strategy of every bidder is a function of his own valuation $v_i$, denoted by $B_i(v_i)$. $B_i$ maps a valuation function $v_i\in V_i$ to a {\em distribution} $B_i(v_i)=B_i^{v_i}$, over all possible bid vectors (strategies) for $i$. In this case we will write $\vec{b}_i\sim B_i^{v_i}$, for any particular bid vector $\vec{b}_i$ drawn from this distribution. We also use the notation $\vec{B}_{-i}^{\vec{v}_{-i}}$, to refer to the vector of randomized strategies of bidders other than $i$, under valuation profile $\vec{v}_{-i}$ for these bidders. A {\it Bayes-Nash equilibrium} (BNE) is a strategy profile $\vec{B}=(B_1,\dots, B_n)$ such that, for every bidder $i$ and for every valuation $v_i$, $B_i(v_i)$ maximizes the utility of $i$ in expectation, over the distribution of the other bidders' valuation functions $\vec{w}_{-i}$ {\em given $v_i$} and over the distribution of $i$'s own and the other bidders' strategies, $\vec{B}^{(v_i,\vec{w}_{-i})}$. That is, for every pure strategy $\vec{c}_i$ of $i$:

\[
\mathbb{E}_{\vec{w}_{-i}|v_i}\Biggl[
\mathbb{E}_{\vec{b}\sim\vec{B}^{(v_i,\vec{w}_{-i})}}\Bigl[
u_i(\vec{b})
\Bigr]\Biggr]\geq 
\mathbb{E}_{\vec{w}_{-i}|v_i}\Biggl[
\mathbb{E}_{\vec{b}_{-i}\sim\vec{B}^{\vec{w}_{-i}}}\Bigl[
u_i(\vec{c}_i,\vec{b}_{-i})
\Bigr]\Biggr]
\]

\noindent where we use the notation $\mathbb{E}_{\vec{v}}$ and $\mathbb{E}_{\vec{w}_{-i}|v_i}$ to denote the expectation over the distribution $\pi$ and over $\pi(\cdot|v_i)$ respectively, i.e., {\em given} $v_i$ (instead of using $\mathbb{E}_{\vec{v}\sim\pi}$ and $\mathbb{E}_{\vec{w}_{-i}\sim\pi(\cdot|v_i)}$, since the analysis is always in the context of $\pi$ and $\pi_i$).

\noindent Fix a valuation profile $\vec{v}\in{\cal V}$ and consider a (mixed) bidding configuration $\vec{B}^{\vec{v}}$, under $\vec{v}$. The Social Welfare $SW(\vec{B}^{\vec{v}})$ under $\vec{B}^{\vec{v}}$ is defined in expectation over the bidding profiles chosen by the bidders collectively, from their randomized strategies:

\[
SW(\vec{B}^{\vec{v}})=\mathbb{E}_{\vec{b}\sim\vec{B}^{\vec{v}}}\left[\sum_iv_i(x_i(\vec{b}))\right]
\]

\noindent The {\em expected} Social Welfare in {\em Bayes-Nash Equilibrium} is $\mathbb{E}_{\vec{v}\sim\pi}[SW(\vec{B}^{\vec{v}})]$. The socially optimal assignment under valuation profile $\vec{v}\in{\cal V}$ will be denoted by $\vec{x}^{\vec{v}}$. The {\em expected} optimal social welfare is then $\mathbb{E}_{\vec{v}}[SW(\vec{x}^{\vec{v}})]$, by slight abuse of notation. Under these definitions, we will be studying the {\em Bayes-Nash Price of Anarchy}, i.e., the worst case ratio
$\mathbb{E}_{\vec{v}}[SW(\vec{x}^{\vec{v}})] /
\mathbb{E}_{\vec{v}}[SW(\vec{B}^{\vec{v}})]$ over all possible distributions $\pi$ and Bayes-Nash equilibria $\vec{B}$.

As in previous works~\cite{Christodoulou08,Feldman13}, we ensure existence of Bayes-Nash equilibria in our auction format by assuming that bidders have bounded and finite strategy spaces, e.g., derived through discretization. Our bounds on the Bayesian inefficiency hold for sufficiently fine discretizations (see also the discussion in Appendix D of~\cite{Feldman13}).

\section{Undominated Equilibria}
\label{section:undominated-pne}

In this work we consider only bidders with submodular valuation functions, so that $m_i(1)\geq ...\geq m_i(k)$, for every bidder $i$. As already explained in Section~\ref{section:model-definitions}, the bidding interface of the Uniform Price Auction requires that each bidder submits a sequence of non-increasing marginal bids $\vec{b}_i=(b_i(1),\dots,b_i(k))$, with $b_i(1)\geq b_i(2)\geq\cdots\geq b_i(k)$ (see e.g., the related chapters in~\cite{Krishna02} and~\cite{Milgrom04}). By the auction's definition, it follows that, under any strategy profile $\vec{b}$, the uniform price $p(\vec{b})$ never exceeds any of the winning (marginal) bids.

Lemmas~\ref{lm:no-overbidding} and~\ref{lm:vi(1)} below state two well known facts about the Uniform Price Auction with submodular bidders (see e.g.,~\cite{Krishna02,Milgrom04}). For the sake of clarity and completeness, we state them and prove them here to clarify some ambiguities and emphasize their dependence on the requirement that bidders issue non-increasing marginal bids. Their proofs are provided in {\bf Appendix A}.

\begin{lemma}
For bidders with submodular valuation functions, and for any $j\in [k]$, it is a weakly dominated strategy to declare a bid $b_i(j)$ with $b_i(j) > m_i(j)$.
\label{lm:no-overbidding}
\end{lemma}

An assumption that is recently being used in various other auction formats is that bidders do not overbid their value for a set of goods (e.g., \cite{Bhawalkar11,CKK+12,Christodoulou08,Leme10}). The justification for this is that such strategies may be dominated by other strategies and hence should be avoided. In our context, this would mean that for any {\em number} of $r$ units, $\sum_{j=1}^r b_i(j) \leq v_i(r) = \sum_{j=1}^r m_i(j)$. We note here that Lemma \ref{lm:no-overbidding} shows that a weakly undominated strategy in our setting implies a stricter notion of conservative behavior than the usual ``no-overbidding'' assumption of the literature. To distinguish from the usual no-overbidding assumption, we call a bidder $i$ who does not bid beyond $m_i(j)$ for any $j\in[k]$ {\it conservative in marginal values}.

\begin{lemma}
In an undominated strategy, a bidder with a submodular valuation, never declares a bid $b_i(1)\neq v_i(1)$.
\label{lm:vi(1)}
\end{lemma}

We now give a characterization of a subset of undominated pure Nash equilibria, which will be utilized for the analysis of their social inefficiency in Section~\ref{section:pne-poa}.

\begin{lemma}
Let $\vec{b}$ be a pure Nash equilibrium strategy profile of the Uniform Price
Auction in undominated strategies for submodular bidders, with
uniform price $p(\vec{b})$. There always exists a pure Nash equilibrium
$\vec{b}'$ in undominated strategies, satisfying $\vec{x}(\vec{b}') =
\vec{x}(\vec{b})$ and:
\begin{enumerate}
\item $b_i'(x) = m_i(x)$, for every bidder $i$ and every $x\leq x_i(\vec{b})$.
\item $p(\vec{b}')\leq p(\vec{b})$ and $p(\vec{b}')$ is either $0$ or equal to
$v_i(1)$ for some bidder $i$.
\end{enumerate}
\label{lemma:undominated-pne}
\end{lemma}

\begin{proof}
Since $\vec{b}$ is an equilibrium in undominated strategies, for every bidder $i$ we have $b_i(1) = v_i(1)$. Also, $b_i(x)\leq m_i(x)$ for every $x\geq 2$. For every winning marginal bid $b_i(x)$ of any bidder $i$ and for $x=2,\dots,x_i(\vec{b})$, we may set $b_i'(x)=m_i(x)$ without altering neither the uniform price nor the units obtained by the winners under $\vec{b}$. This is because the losing bids remain unchanged whereas the winning bids only increase, hence the price is the same as before. 
Finally, the new profile is easily shown to be an equilibrium simply because we started with an equilibrium $\vec{b}$. Thus the first condition $b_i'(x)=m_i(x)$ for $x\leq x_i(\vec{b})$ can be satisfied for every winner $i$.

For the second condition, consider an equilibrium $\vec{b}$ satisfying the first condition. Assume that $\vec{b}$ is such that $p(\vec{b})=b_i(j)$, for some $j\geq 1$. If $j=1$, then it must be $x_i(\vec{b})=0$ and then $p(\vec{b}) = b_i(1) = v_i(1)$, because $\vec{b}$ is undominated. Consider the case where $j\geq 2$. Then $x_i(\vec{b})\geq 1$. If the value $b_i(j)=p(\vec{b})$ is {\em unique} in $\vec{b}$, then either $p(\vec{b})=0$, or $i$ could lower all his marginal bids for the $j$-th unit and onwards to $0$. This way, he would strictly increase his utility, by lowering the uniform price to the next highest marginal bid, which contradicts that $\vec{b}$ is a PNE. Thus we may assume that more than one highest losing bids of the same value exist, that determine the price. If none of them are equal to $v_{i}(1)$ for some $i$, we can zero out the bids of all bidders who have the highest losing bid (from that bid onwards) and again obtain a bidding configuration $\vec{b}'$ with a price $p(\vec{b}')<p(\vec{b})$. The new vector $\vec{b}'$ is still an undominated PNE. The utilities of all winners have now increased, but none of them may increase his utility more by deviating unilaterally, because we have not altered any of the winning bids. It is the last winning bid that will determine the new uniform price in case of any bidder trying to win an extra unit by deviating. But since there was no incentive to do such a deviation under $\vec{b}$ the same is true for $\vec{b'}$ too. Hence, we have managed to reduce the price and have some bidders zero out their losing bids. We can repeat this procedure for the configuration $\vec{b}'$, until we reach a a configuration satisfying the second property.\qed
\end{proof}

\section{Inefficiency of Undominated Pure Nash Equilibria}
\label{section:pne-poa}

We develop a welfare guarantee for pure Nash equilibria in undominated strategies, of the Uniform Price Auction, for bidders with submodular valuation functions. Recall that, given a configuration $\vec{b}$, we denote the ($k$ highest) winning bids by $\beta_j(\vec{b})$, $j=1,\dots,k$, so that $\beta_1(\vec{b})\leq\beta_2(\vec{b})\leq\cdots\leq\beta_k(\vec{b})$. We extend this notation to partial configurations $\vec{b}_{-i}$ for any bidder $i\in[n]$. 

Let $\vec{x}^*=(x_1^*,\dots,x_n^*)$ be a socially optimal allocation and $\vec{b}$ be a bidding configuration corresponding to an undominated pure Nash equilibrium of the auction. Under $\vec{b}$ the allocation is $\vec{x}(\vec{b})=(x_1(\vec{b}),\dots,x_n(\vec{b}))$. To simplify notation we use $\vec{x}$ for $\vec{x}(\vec{b})$ and $x_i$ for $x_i(\vec{b})$. Given any allocation $\vec{x}$, define the set ${\cal W}(\vec{x})=\{i|x_i\geq 1\}$ to be the subset of {\em winners}, i.e. bidders that receive at least one item unit. We also define $3$ additional sets, ${\cal W}_0(\vec{x}),{\cal W}_1(\vec{x}),{\cal W}_2(\vec{x})$, all with reference to $\vec{x}^*$ as follows:

\[
{\cal W}_0(\vec{x})=\{i\in{\cal W}(\vec{x}^*)|x_i\geq x_i^*\}\mbox{,\quad}
{\cal W}_1(\vec{x})=\{i\in{\cal W}(\vec{x}^*)|x_i < x_i^*\},
\]

\noindent and ${\cal W}_2(\vec{x})={\cal W}(\vec{x})\setminus
\Bigl({\cal W}_0(\vec{x})\cup{\cal W}_1(\vec{x})\Bigr)$.

We note that given any assignment $\vec{x}=\vec{x}(\vec{b})$ for a pure Nash equilibrium configuration $\vec{b}$, ${\cal W}_0(\vec{x})\cup{\cal W}_1(\vec{x})\cup{\cal W}_2(\vec{x})$ is a partition of ${\cal W}(\vec{x})$, given the assumption of undominated strategies about $\vec{b}$; every winner $i\in{\cal W}(\vec{x}^*)$ will still be a winner also under $\vec{b}$, because of specifying his $v_i(1)$ truthfully (by Lemma \ref{lm:vi(1)}), thus obtaining at least one unit. Hence none of ${\cal W}_0(\vec{x})$, ${\cal W}_1(\vec{x})$, ${\cal W}_2(\vec{x})$ may contain non-winning bidders of $\vec{x}$. First we present a general upper bound on the Price of Anarchy for undominated pure Nash equilibria. 

\begin{lemma}
Let $\vec{b}$ denote any undominated pure Nash equilibrium of a Uniform Price Auction for $k$ units and $\vec{x}^*$ be an assignment that maximizes the social welfare. The Price of Anarchy is:
\begin{equation}
PoA\leq\sup_{\vec{b}}\max_{i:x_i^*-x_i(\vec{b})>0}\left[
v_i(x_i^*)\cdot\left(
v_i\Bigl(x_i(\vec{b})\Bigr)+\sum_{j=1}^{x_i^*-x_i(\vec{b})}\beta_j(\vec{b})
\right)^{-1}\right]
\label{equation:poa-general-ub}
\end{equation}
\label{lemma:poa-general-ub}
\end{lemma}
The proof of the lemma is given in {\bf Appendix B}. We can now present our constant bound on the Price of Anarchy:

\begin{theorem} 
The Uniform Price Auction recovers in undominated pure Nash equilibrium a fraction of at least $1-e^{-1}$ of the optimal Social Welfare, for multi-demand bidders with symmetric submodular valuation functions.
\label{theorem:multi-poa}
\end{theorem}

\begin{proof}
Without loss of generality, it suffices to upper bound the social inefficiency of undominated equilibria $\vec{b}$ satisfying the properties of Lemma~\ref{lemma:undominated-pne}. Let $p(\vec{b})$ be the uniform price payed under equilibrium $\vec{b}$, i.e. the value of the highest losing (marginal) bid. In order to estimate a lower bound on the Social Welfare of $\vec{b}$, we consider possible deviations of bidders $i\in{\cal W}_1(\vec{x}(\vec{b}))$. We may assume that ${\cal W}_1(\vec{x}(\vec{b}))\neq\emptyset$ for, otherwise, ${\cal W}_0(\vec{x}(\vec{b}))={\cal W}(\vec{x}^*)$ and $\vec{b}$ is socially optimal, i.e. $SW(\vec{b})=SW(\vec{x}^*)$.

For every bidder $i\in{\cal W}_1(\vec{x}(\vec{b}))$, define $r_i(\vec{b})=x_i^*-x_i(\vec{b})$. For every bidder $i\in{\cal W}_1(\vec{x}(\vec{b}))$ and for every value $j=1,\dots,r_i(\vec{b})$, there exists a deviation that will grant him $j$ more units, additionally to the ones he already wins under $\vec{b}$; this is justified by the fact that at equilibrium $\vec{b}$, all bidders play marginal bids {\em equal} to their actual marginal values, or $0$ (as prescribed by the properties given in Lemma~\ref{lemma:undominated-pne}). Since the optimal assignment results from a simple sorting of the actual marginal values, every winner $i\in{\cal W}(\vec{x}^*)$ may feasibly deviate under $\vec{b}$, so as to obtain any total number of units between $x_i$ and $x_i^*$. A deviation of $i\in{\cal W}_1(\vec{x}(\vec{b}))$ for obtaining any number of $j=1,\dots,r_i(\vec{b})$ {\em additional} units will raise the uniform price to exactly $\beta_j(\vec{b})$ and cannot be profitable for $i$, i.e.:

\[
v_i(x_i(\vec{b})+j) - (x_i(\vec{b})+j)\cdot\beta_j(\vec{b})
\leq 
v_i(x_i(\vec{b})) - x_i(\vec{b})\cdot p(\vec{b})
\]

\noindent To simplify notation, we use hereafter $x_i$ for $x_i(\vec{b})$, $p$ for $p(\vec{b})$, $r_i$ for $r_i(\vec{b})$, and $\beta_j$ for $\beta_j(\vec{b})$, (always with respect to an undominated pure Nash equilibrium $\vec{b}$).

Then we deduce that for every $i\in{\cal W}(\vec{x}^*)$:

\begin{equation}
\beta_j\geq\frac{1}{j+x_i}\cdot\Bigl(v_i(x_i+j)-v_i(x_i)+x_i\cdot p\Bigr)
\mbox{,\ \ for\ }
j=1,\dots,r_i
\label{equation:equilibrium}
\end{equation}

\noindent We can now proceed to upper bound~\eqref{equation:poa-general-ub} from Lemma~\ref{lemma:poa-general-ub}, using~\eqref{equation:equilibrium} as follows:

\begin{align}
&
\displaystyle 
v_i(x_i)+\sum_{j=1}^{r_i}\beta_j
\geq \displaystyle
v_i(x_i)+\sum_{j=1}^{r_i}\frac{1}{j+x_i}\cdot
\Bigl(v_i(x_i+j)-v_i(x_i)\Bigr)
\label{equation:bids-lb}\\
& = \displaystyle
v_i(x_i)+\sum_{j=1}^{r_i}\left(
\frac{j}{j+x_i}\cdot\frac{v_i(x_i+j)- v_i(x_i)}{j}
\right)\nonumber\\
& \geq \displaystyle
v_i(x_i)+\frac{v_i(x^*_i)-v_i(x_i)}{x_i^*-x_i}\cdot
\sum_{j=1}^{r_i}\frac{j}{j+x_i}
\label{equation:submod-1}\\
& = \displaystyle
v_i(x_i)+\frac{v_i(x^*_i)-v_i(x_i)}{x^*_i-x_i}\cdot\left(
x_i^*-x_i-x_i\cdot\sum_{j=1}^{r_i}\frac{1}{j+x_i}
\right)\nonumber\\
& = \displaystyle
v_i(x_i^*)-\frac{v_i(x^*_i)-v_i(x_i)}{x^*_i-x_i}\cdot
x_i\cdot\sum_{j=1}^{r_i}\frac{1}{j+x_i}\\
& \geq \displaystyle
\left(
1-\frac{x_i}{x_i^*}\cdot\sum_{j=1}^{r_i}\frac{1}{j+x_i}\right)\cdot 
v_i(x_i^*)\geq
\left(
1-\frac{x_i}{x_i^*}\cdot\int_{x_i}^{x_i^*}\frac{1}{y}dy\right)\cdot 
v_i(x_i^*)\label{equation:submod-3}\\
& = \displaystyle
\left(
1+\frac{x_i}{x_i^*}\cdot\ln\frac{x_i}{x_i^*}\right)\cdot
v_i(x_i^*)\geq (1-e^{-1})\cdot v_i(x_i^*)\label{equation:end-result}
\end{align}

\noindent \eqref{equation:bids-lb} occurs by substitution of $\beta_j$ from \eqref{equation:equilibrium} and after dropping the $x_i\cdot p\geq 0$ term. \eqref{equation:submod-1} follows by submodularity of the valuation functions, particularly that $\frac{v_i(x_i+j)- v_i(x_i)}{j}\geq\frac{v_i(x_i^*)-v_i(x_i)}{x_i^*-x_i}$, for any $j=1,\dots,r_i$ where $r_i=x_i^*-x_i$. For \eqref{equation:submod-3} we used $\frac{v_i(x_i^*)-v_i(x_i)}{x_i^*-x_i}\leq\frac{v_i(x_i^*)}{x_i^*}$, given $v_i(0)=0$; we bounded the sum of harmonic terms with the integral, using $\sum_{k=m}^nf(k)\leq\int_{m-1}^nf(x)dx$, for a monotonically decreasing positive function. We obtain the final result by minimizing $f(y)=1+y\ln y$ over $(0,1)$ for $y=e^{-1}$. The claimed bound for the Price of Anarchy follows by Lemma~\ref{lemma:poa-general-ub}.\qed
\end{proof}

\bigskip

\noindent We will produce an almost matching lower bound for the result of theorem~\ref{theorem:multi-poa}, which holds for any number of units $k\geq 9$. We pause here to discuss first three simple tight examples for $k=2,3,4$ units.

\medskip

\noindent {\bf Examples.} We give a detailed example for $k=3$ first. We show a simple lower bound of $\frac{18}{13}$. Consider $3$ item units and $3$ bidders, with valuation functions $v_1(x)=x$, $v_2(x)=\frac{2}{3}$, $v_3(x)=\frac{1}{2}$. The socially optimal welfare is $3$, where bidder $1$ gets all the units. Consider the pure Nash equilibrium configuration $\vec{b}_1=(1,0,0)$, $\vec{b}_2=(\frac{2}{3},0,0)$ and $\vec{b}_3=(\frac{1}{2},0,0)$, for $j=1,2,3$. None of bidders $2$ and $3$ have incentive to change their bid, as their utility cannot improve. Should bidder $1$ raise $b_1(2)=0$ to $b_1'(2)>\frac{1}{2}$, he would win one more unit additionally to $x_1(\vec{b})=1$, but pay $x_1(\vec{b}')\times\frac{1}{2}=1$, thus not improving his utility. If bidder $1$ would raise $b_1(2)=0$ to $b_1'(2)>\frac{2}{3}$ and $b_1(3)=0$ to $b_1'(3)>\frac{1}{2}$, he would still not improve his utility; it would be $u_1(\vec{b}')=v_1(x_i(\vec{b}'))-x_i(\vec{b}')\times\frac{2}{3} = 3-3\times\frac{2}{3}=1=u_1(\vec{b})$. The Price of Anarchy in this example is $\frac{18}{13}$.

For $k=2$ and $k=4$ our examples follow a similar pattern to the one discussed for $k=3$. We can take $n=k=2$ bidders with valuation functions $v_1(x)=x$ and $v_2(x)=\frac{1}{2}$. This instance yields a Price of Anarchy at least $\frac{4}{3}$; in the social optimum, bidder $1$ gets both units, while in the (undominated) equilibrium profile, $\vec{b}_1=(1,0)$, $\vec{b}_2=(\frac{1}{2},0)$ (and $SW(\vec{b})=\frac{3}{2}$). For $k=4$, we take $n=k=4$ bidders with valuation functions $v_1(x)=x$, $v_2(x)=\frac{1}{2}$, $v_3(x)=\frac{2}{3}$, $v_4(x)=\frac{3}{4}$. In the social optimum we obtain welfare $4$, by giving all units to bidder $1$. For an undominated equilibrium profile, take $\vec{b}_1=(1,0,0,0)$, $\vec{b}_2=(\frac{1}{2},0,0,0)$, $\vec{b}_3=(\frac{2}{3},0,0,0)$, $\vec{b}_4=(\frac{3}{4},0,0,0)$. Then, $SW(\vec{b})=35/12$ and this gives a lower bound of $48/35$ for the Price of Anarchy. 

Notice that our lower bound for $k=4$ is smaller than for $k=3$. None the less, all these examples can be shown to provide tight lower bounds for their corresponding cases $k=2,3,4$, by usage of simple arguments and explicit treatment of the left-hand side of~\eqref{equation:submod-3}. For $k=2$, our argument is that the worst-case instance occurs when there is a {\em single} winner in the social optimum; otherwise, if there are $2$ winners, they also remain winners in undominated pure Nash equilibrium (due to Lemma~\ref{lm:vi(1)}) and the Price of Anarchy is $1$. In the worst-case instance though, the bidder being a winner (of $2$ units) in the social optimum remains a winner (of exactly $1$ unit) in the undominated pure Nash equilibrium. Taking $x_i=1$, $x_i^*=2$, $r_i=1$ in the left-hand side of~\eqref{equation:submod-3}, sufficies to obtain $\frac{3}{4}v_i(x_i^*)$. 

For the case of $k=3$ the reasoning is similar. There have to be strictly less than $3$ winners in the social optimum for, otherwise, the allocation coincides with that of an undominated pure Nash equilibrium. Now, if there are $2$ winners, one of them obtains only $1$ unit (out of $k=3$) in both, socially optimal allocation and undominated equilibrium allocation. Thus, the social inefficiency is due to the other winner losing one unit in equilibrium (of the $2$ that he obtains in the social optimum). By our experience with $k=2$(and by Lemma~\ref{lemma:poa-general-ub}), this example cannot have Price of Anarchy more than $4/3$. Thus, we may assume that there exists a single winner in the social optimum. To achieve maximum welfare ``damage'' in undominated pure Nash equilibrium, this single winner of $3$ units loses $2$ of them in equilibrium and we apply~\eqref{equation:submod-3} appropriately, to obtain an upper bound of $18/13$. For $k=4$ the reasoning uses our experience from both previous cases. We exclude instances with $4$, $3$ and $2$ winners in the social optimum, as they cannot have Price of Anarchy more than $1$, $4/3$, $18/13$ respectively. For maximum welfare damage, we assume that the single winner loses $3$ out of $4$ units in undominated pure Nash equilibrium.

\medskip

The following more general lower bound is valid for at least $k\geq 9$ units. For the remaining values of $k=5,6,7,8$, we do not have any tighter upper and lower bounds.

\begin{theorem}
For any $k\geq 9$, the Uniform Price Auction recovers in undominated pure Nash equilibrium at most a factor $(1-e^{-1}+\frac{2}{k})$ of the optimal social welfare, even for $2$ submodular bidders.
\end{theorem}

\begin{proof}
Consider $k\geq 9$ units and $2$ bidders. For $q=\lfloor e^{-1}\cdot k-1\rfloor$ (notice that $q\geq 1$) define the valuation functions to be:

\[
v_1(x) = x
\mbox{\quad\quad and\quad\quad}
v_2(x)=\left\{
\begin{array}{ll}
x-q\cdot(H_k-H_{k-x})\, & x\leq k-q\medskip\\
k-q\cdot(1+H_k-H_q)\, & x > k-q
\end{array}\right.
\]

\noindent where $H_m$ is the $m$-th harmonic number. Notice that the marginal values of bidder $2$ are equal to $0$ for $x>k-q$. It can be verified that $v_2$ is symmetric submodular in $x$:

\[
v_2(x)=x-q\cdot\Bigl(H_k-H_{k-x}\Bigr)=
\sum_{j=1}^x\left(1-\frac{q}{k-j+1}\right)=
\sum_{j=1}^x\frac{r-j+1}{k-j+1}
\]

\noindent where $r=k-q$. Then
$\frac{r-j+1}{k-j+1}\leq\frac{r-j+2}{k-j+2}=\frac{r-(j-1)+1}{k-(j-1)+1}$, thus
$v_2(x)-v_2(x-1)\leq v_2(x-1)-v_2(x-2)$, for $x\leq k-q$; for $x > k-q$,
$v_2(x)=v_2(x-1)$, thus $v_2$ is submodular.

For the optimal social welfare we grant all units to bidder $1$, i.e. $\vec{x}=(k,0,\dots,0)$ and obtain a total welfare $SW(\vec{x}^*)=k$. For the equilibrium configuration $\vec{b}$ we set:

\[
b_1(j)=\left\{\displaystyle
\begin{array}{cl}
1, & \mbox{for\ }j\leq q\smallskip\\
0, & \mbox{for\ }j > q
\end{array}
\right.
\mbox{\quad\quad}
b_2(j)=\left\{\displaystyle
\begin{array}{cl}
\frac{r-j+1}{k-j+1}, & \mbox{for\ }j\leq r=k-q\smallskip\\
0, & \mbox{for\ }j >r
\end{array}
\right.
\]
 
Thus, under $\vec{b}$, $q$ units are obtained by bidder $1$ and $k-q$ units by bidder $2$. We show that $\vec{b}$ is a pure Nash equilibrium. Notice that bidder $2$ is essentially truthful in this profile and may not increase his bids further so as to obtain another unit (given that he plays undominated strategies). On the other hand, the uniform price is $0$ in this setting, so bidder $2$ does obtain the maximum of his utility for the won units. Bidder $1$ also pays the uniform price of $0$, so he does not have incentive to drop any of his units. Should bidder $1$ try to retain any $j\leq r$ of the $r=k-q$ units held by bidder $2$, the uniform price would become $\frac{j}{k-r+j}$ and bidder $1$ will hold a total of $k-r+j$ units. The marginal gain from bidder $1$ obtaining the extra $j$ units is cancelled out by a total payment equal to $j$; thus bidder $1$ does not have incentive to deviate under $\vec{b}$.

For the social welfare of $\vec{b}$ we have:

\[
\begin{array}{lclcl}
SW(\vec{b}) & = & v_1(q)+v_2(r) & = &  q+r-q\cdot\Bigl(H_k-H_q\Bigr)\\
& = &  k - q\cdot\Bigl(H_k-H_q\Bigr) & = & k\cdot\Bigl(1-\frac{q}{k}\cdot\left(H_k-H_q\right)\Bigr)
\end{array}
\]

\noindent Then, the Price of Anarchy is at least:

\begin{align}
\displaystyle
\frac{k}{
k\cdot\Bigl(1-\frac{q}{k}\cdot\left(H_k-H_q\right)\Bigr)} 
&=\displaystyle
\left(1-\frac{q}{k}\cdot\Bigl(H_k-H_q\Bigr)\right)^{-1}\nonumber\\
\geq\displaystyle
\left(1-\frac{e^{-1}\cdot k-2}{k}\cdot\int_{q+1}^{k}\frac{1}{y}dy
\right)^{-1}
& =\displaystyle
\left(1-\frac{e^{-1}\cdot k-2}{k}\cdot\ln\frac{k}{\lfloor e^{-1}k-1\rfloor+1}
\right)^{-1}\nonumber\\
\geq\displaystyle
\left(1-e^{-1}+\frac{2}{k}\right)^{-1} & \nonumber
\end{align}

\noindent where we used $H_k-H_q=\sum_{r=q+1}^k\frac{1}{r}\geq\int_{q+1}^{k+1}\frac{1}{y}dy\geq\int_{q+1}^k\frac{1}{y}dy$, for monotonically decreasing positive functions; the final derivation follows by $q+1\leq e^{-1}\cdot k$ and $\lfloor e^{-1}k-1\rfloor+1\geq e^{-1}k$\qed
\end{proof}

\section{Bayes-Nash Inefficiency with Undominated Support}
\label{section:bne-poa}

In this section we investigate the social inefficiency of (mixed) Bayes-Nash equilibria for bidders with submodular valuation functions. Just as in the case of pure equilibria we focused on undominated strategies, here we will focus on mixed Bayes-Nash equilibria that are supported by pure undominated strategies. We wish to note, however, that the assumption of undominated support is only marginally restrictive for the main (social inefficiency) result described in this section; it guarantees the properties given by Lemmas~\ref{lm:no-overbidding} and~\ref{lm:vi(1)}, that allow us to prove a tighter inefficiency bound. In the end of this section we discuss how our analysis leads to a similar (slightly worse) bound, under the standard assumption of no-overbidding used, e.g., in~\cite{Christodoulou08,Bhawalkar11,Feldman13}.

Following~\cite{Christodoulou08,Bhawalkar11}, to ensure the existence of mixed Bayes-Nash equilibria, we make the assumption of a finite bidding space for bidders under sufficiently fine discretization.
We claim that such mixed Bayes-Nash equilibria supported by undominated pure strategies exist for the case of bidders with submodular valuation functions. Indeed, consider a Bayesian game where the strategy space of each player is bounded (e.g. of the form $[0, U]$, where $U$ is a sufficiently large upper bound on the values of all bidders) and finite, through some sufficiently fine discretization. 

We first claim that in our setting there always exists a Bayes-Nash equilibrium where all strategies used in its support are not weakly dominated. Then we will claim that these strategies are conservative w.r.t. marginal bids. To argue about these statements, we use two well known facts from game theory. The first one is that a Bayes-Nash equilibrium $\vec{B}$ can be seen as a Nash equilibrium of a complete information game (see e.g. \cite{Osborne09}[Chapter 9]), where the set of players is the set of pairs $(i, v_i)$ for every $i=1,...,n$ and $v_i\in V_i$, and the strategy space of player $(i, v_i)$ is the same as the strategy space of $i$ in the Bayesian game. For a given mixed strategy profile $\vec{B}$ in this game, the utility function of player $(i, v_i)$, denoted by $u_i^{v_i} (\vec{B})$, is:

$$
u_i^{v_i}(\vec{B}) = 
\mathbb{E}_{\vec{v}_{-i}|v_i}\Biggl[
\mathbb{E}_{\vec{b}\sim\vec{B}^{(v_i,\vec{v}_{-i})}}
\Bigl[u_i(\vec{b})\Bigr]\Biggr]
=\sum_{\vec{v}_{-i}\in{\cal V}_{-i}}\pi(\vec{v}_{-i})
\mathbb{E}_{\vec{b}\sim\vec{B}^{(v_i,\vec{v}_{-i})}}\Bigl[
u_i(\vec{b})\Bigr]
$$
 
Note that the utility of a player $(i, v_i)$ does not depend on the action of other pairs that involve $i$, (i.e., on the other types of player $i$). Since this complete-information game is a finite game it possesses a mixed Nash equilibrium. It is well known that in any finite game there always exists a mixed equilibrium where no weakly dominated action is contained in the support of each player's strategy (see \cite{Osborne09} [Section 4.4]). Putting everything together, we have that there always exists a Bayes-Nash equilibrium where every strategy used in its support is not weakly dominated. It is an easy exercise to verify the validity of Lemma~\ref{lm:no-overbidding} for a sufficiently fine discretization of the strategy space and that Lemma~\ref{lm:vi(1)} holds, for any such discretization that does not exclude the actual marginal values from the bidders' strategy space.

We introduce some auxiliary notation for the analysis that follows. For any valuation profile $\vec{v}\in{\cal V}$ let $\vec{x}^{\vec{v}}=(x_1^{\vec{v}},\dots,x_n^{\vec{v}})$ denote the socially optimal assignment. For any particular bidder $i\in[n]$ let ${\cal U}^i\subseteq{\cal V}$ denote the subset of valuation profiles $\vec{v}\in{\cal V}$ where $x_i^{\vec{v}}\geq 1$, i.e., ${\cal U}^i=\{\vec{v}\in{\cal V}|x_i^{\vec{v}}\geq 1\}$; these are the profiles under which $i$ is a ``social optimum winner''. Accordingly, we let ${\cal W}^{\vec{v}}$ denote the subset of ``social optimum winners'' in valuation profile $\vec{v}\in{\cal V}$. Given any (pure) bidding profile $\vec{b}$, we use the {\em ``operator''} $\beta_j(\vec{b})$ here as well, to denote the $j$-th lowest winning bid in $\vec{b}$, as in Section~\ref{section:pne-poa}. The following Lemma will serve the purpose of lower bounding the Social Welfare of a profile $\vec{b}$ by a sum of the winning bids (much like inequalities~\eqref{equation:first-term} and~\eqref{equation:third-term} along with Lemma~\ref{lm:no-overbidding} did, in the proof of Lemma~\ref{lemma:poa-general-ub} in the previous section).

\begin{lemma}
For a valuation profile $\vec{v}$, let $\vec{b}$ denote an arbitrary pure bidding profile in undominated strategies, $p(\vec{b})$ be the uniform price under $\vec{b}$ and let $\vec{x}^*$ be the efficient (socially optimal) assignment of $k$ units to $\ell\leq k$ winning bidders w.r.t. $\vec{v}$. Fix an arbitrary ordering of the $\ell$ winning bidders under $\vec{x}^*$ and define $t_i=\lceil\frac{x_i^*}{2}\rceil$, $i=1,\dots,\ell$. Then: 
\begin{equation}
\sum_{1\leq i\leq\ell}t_i\beta_{t_i}(\vec{b}_{-i})
\leq
p(\vec{b})+\frac{k-1}{k}SW(\vec{b})
\label{equation:sw-cover}
\end{equation}
\label{lemma:sw-cover}
\end{lemma}
To ease our way towards the main result of this section, we defer the technical proof of this Lemma to the end of our exposition, along with a discussion on how to replace the undominated support assumption with the no-overbidding assumption.

The following Lemma facilitates the expression of BNE conditions regarding
unilateral deviations, and has been proved in a different form and under a
different context (for simultaneous single-unit auctions with combinatorial
bidders) also in~\cite{Christodoulou08,Bhawalkar11}. We provide its proof here for completeness.
\begin{lemma}
For every bidder $i\in[n]$ with submodular valuation $v_i$ define the
bidding vector $\vec{m}_i^{[j]}=(m_i(1),m_i(2),\dots,m_i(j),0,0,\dots,0)$ .
For any conservative bidding profile $\vec{b}_{-i}$, and for any number of
units $j\in\{0, 1,\ldots, k\}$: 
$$
u_i(\vec{m}_i^{[j]},\vec{b}_{-i})\geq
v_i(j)-j\cdot\beta_j(\vec{b}_{-i}).
$$
\label{lemma:bayesian-lemma}
\end{lemma}

\noindent The proof of this lemma for our setting is given in {\bf Appendix C}. Now we can show the main result of this section:

\begin{theorem}
The Price of Anarchy of Bayes-Nash Equilibria with undominated support in Uniform Price Auctions for bidders with submodular valuation functions is at most $4-\frac{2}{k}$.
\label{theorem:bne-poa}
\end{theorem}

\begin{proof}
Consider a Bayes-Nash equilibrium $\vec{B}$. Fix a bidder $i$ and any valuation profile $\vec{v}=(v_i,\vec{v}_{-i})\in{\cal V}$. For $t_i^{\vec{v}}=\left\lceil\frac{x_i^{\vec{v}}}{2}\right\rceil$, and for any valuation profile $\vec{w}_{-i}\in{\cal V}_{-i}$ and strategy $\vec{b}_{-i}\sim\vec{B}_{-i}^{\vec{w}_{-i}}$, we apply Lemma~\ref{lemma:bayesian-lemma}. Then, we take the expectation over the randomized strategies of the other bidders and, subsequently, over all valuation profiles $\vec{w}_{-i}\in{\cal V}_{-i}$, to obtain:

\begin{align}
\displaystyle
\mathbb{E}_{\vec{w}_{-i}|v_i}\Bigl[
\mathbb{E}_{\vec{b}_{-i}\sim\vec{B}_{-i}^{\vec{w_{-i}}}}[
u_i(\vec{m}_i^{[t_i^{\vec{v}}]}, \vec{b}_{-i})]\Bigr]
&\displaystyle\geq
v_i(t_i^{\vec{v}})-t_i^{\vec{v}}\cdot
\mathbb{E}_{\vec{w}_{-i}|v_i}\Bigl[
\mathbb{E}_{\vec{b}_{-i}\sim\vec{B}_{-i}^{\vec{w_{-i}}}}[
\beta_{t_i^{\vec{v}}}(\vec{b}_{-i})]\Bigr]\nonumber\\
&\displaystyle\geq
\frac{v_i(x_i^{\vec{v}})}{2}-\mathbb{E}_{\vec{w}}\left[
\mathbb{E}_{\vec{b}\sim\vec{B}^{\vec{w}}}\Bigl[
t_i^{\vec{v}}\cdot\beta_{t_i^{\vec{v}}}(\vec{b}_{-i})\right]\Bigr]
\end{align}

\noindent The second inequality is justified as follows. By independence of the distributions $\{\pi_i\,|\,i\in[n]\}$, we have that for any $i\in[n]$ and any $v_i\in{\cal V}_i$:
\[
\begin{array}{lcl}
\displaystyle
\sum_{\vec{w}_{-i}}\pi(\vec{w}_{-i}|v_i)
&=&
\displaystyle
\sum_{\vec{w}_{-i}}\pi(\vec{w}_{-i})
=
\sum_{\vec{w}_{-i}}\pi(\vec{w}_{-i})\sum_{w_i}\pi(w_i)\\
&=&
\displaystyle
\sum_{(w_i,\vec{w}_{-i})}\pi(\vec{w}_{-i}|w_i)\pi_i(w_i)
=1=\sum_{\vec{w}}\pi(\vec{w})
\end{array}
\]

Also, by submodularity (Proposition~\ref{proposition:valuation-functions}) and monotonicity of valuation functions: $v_i(t_i^{\vec{v}})=v_i(\lceil\frac{x_i^{\vec{v}}}{2}\rceil)\geq\frac{1}{2}v_i(x_i^{\vec{v}})$. Because under BNE $\vec{B}$, bidder $i$ does not have an incentive to deviate, we have:

\[
\mathbb{E}_{\vec{w}_{-i}|v_i}\Bigl[
\mathbb{E}_{\vec{b}\sim\vec{B}^{(v_i,\vec{w}_{-i})}}[
u_i(\vec{b})]\Bigr]\geq
\mathbb{E}_{\vec{w}_{-i}|v_i}\Bigl[
\mathbb{E}_{\vec{b}_{-i}\sim\vec{B}_{-i}^{\vec{w_{-i}}}}[
u_i(\vec{m}_i^{[t_i^{\vec{v}}]}, \vec{b}_{-i})]
\Bigr]
\]

\noindent Thus:

\[
\mathbb{E}_{\vec{w}_{-i}|v_i}\Bigl[
\mathbb{E}_{\vec{b}\sim\vec{B}^{(v_i,\vec{w}_{-i})}}[
u_i(\vec{b})]\Bigr]
+
\mathbb{E}_{\vec{w}}\left[
\mathbb{E}_{\vec{b}\sim\vec{B}^{\vec{w}}}\Bigl[
t_i^{\vec{v}}\cdot\beta_{t_i^{\vec{v}}}(\vec{b}_{-i})\right]\Bigr]
\geq
\frac{v_i(x_i^{\vec{v}})}{2}
\]

\noindent We take expectation of both sides over the distribution of $\vec{v}\in{\cal V}$ and summing over all bidders yields the final expression:

\begin{align}
\displaystyle
\sum_i
\sum_{\vec{v}\in{\cal V}}\pi(\vec{v})\cdot
\mathbb{E}_{\vec{w}_{-i}|v_i}\Biggl[
\mathbb{E}_{\vec{b}\sim\vec{B}^{(v_i,\vec{w}_{-i})}}
\Bigl[u_i(\vec{b})\Bigr]\Biggr]
&\displaystyle+
\sum_{i}
\sum_{\vec{v}\in{\cal V}}\pi(\vec{v})\cdot
\mathbb{E}_{\vec{w}}\Biggl[
\mathbb{E}_{\vec{b}\sim\vec{B}^{\vec{w}}}\Bigl[
t_i^{\vec{v}}\cdot\beta_{t_i^{\vec{v}}}(\vec{b}_{-i})\Bigr]\Biggr]
\nonumber\\
&\displaystyle\geq
\sum_{i}
\sum_{\vec{v}\in{\cal V}}\pi(\vec{v})\cdot
\frac{v_i(x_i^{\vec{v}})}{2}\nonumber\\
&\displaystyle
=
\sum_{\vec{v}\in{\cal V}}\pi(\vec{v})\sum_{i\in{\cal W}^{\vec{v}}}
\frac{v_i(x_i^{\vec{v}})}{2}
=\frac{1}{2}\mathbb{E}_{\vec{v}}[SW(\vec{x^{\vec{v}}})]
\label{equation:bne-guarantee}
\end{align}

\noindent The last equality holds since it is enough to sum over $i\in{\cal W}^{\vec{v}}$, to compute the welfare produced at the optimal assignment with respect to $\vec{v}$. We show in {\bf Appendix C} that the first summand of the left-hand side of~(\ref{equation:bne-guarantee}) satisfies:
\begin{equation}
\sum_i
\sum_{\vec{v}\in{\cal U}^i}\pi(\vec{v})\cdot
\mathbb{E}_{\vec{w}_{-i}|v_i}\Biggl[
\mathbb{E}_{\vec{b}\sim\vec{B}^{(v_i,\vec{w}_{-i})}}
\Bigl[u_i(\vec{b})\Bigr]\Biggr]
=
\mathbb{E}_{\vec{v}}
\Biggl[\mathbb{E}_{\vec{b}\sim\vec{B}^{\vec{v}}}
\left[\sum_iu_i(\vec{b})\right]\Biggr]
\label{equation:left-hand-1}
\end{equation}
Similarly, we have for the second summand on the left-hand side of~\eqref{equation:bne-guarantee}:

\begin{align}
\displaystyle
\sum_{i}
\sum_{\vec{v}\in{\cal U}^i}\pi(\vec{v})\cdot
\mathbb{E}_{\vec{w}}\Biggl[
\mathbb{E}_{\vec{b}\sim\vec{B}^{\vec{w}}}\Bigl[
t_i^{\vec{v}}\cdot\beta_{t_i^{\vec{v}}}(\vec{b}_{-i})\Bigr]\Biggr]
&=\displaystyle
\sum_{i}
\sum_{\vec{v}\in{\cal U}^i}\pi(\vec{v})
\sum_{\vec{w}\in{\cal V}}\pi(\vec{w})\cdot
\mathbb{E}_{\vec{b}\sim\vec{B}^{\vec{w}}}\Bigl[t_i^{\vec{v}}\cdot
\beta_{t_i^{\vec{v}}}(\vec{b}_{-i})\Bigr]\nonumber\\
&=\displaystyle
\sum_{\vec{v}\in{\cal V}}\pi(\vec{v})\sum_{i\in{\cal W}^{\vec{v}}}
\sum_{\vec{w}\in{\cal V}}\pi(\vec{w})\cdot
\mathbb{E}_{\vec{b}\sim\vec{B}^{\vec{w}}}\Bigl[t_i^{\vec{v}}\cdot
\beta_{t_i^{\vec{v}}}(\vec{b}_{-i})\Bigr]\nonumber\\
&=\displaystyle
\sum_{\vec{v}\in{\cal V}}\pi(\vec{v})
\sum_{\vec{w}\in{\cal V}}\pi(\vec{w})\cdot
\mathbb{E}_{\vec{b}\sim\vec{B}^{\vec{w}}}\left[
\sum_{i\in{\cal W}^{\vec{v}}}t_i^{\vec{v}}\cdot
\beta_{t_i^{\vec{v}}}(\vec{b}_{-i})\right]\nonumber\\
&=\displaystyle
\mathbb{E}_{\vec{v}}
\left[\mathbb{E}_{\vec{b}\sim\vec{B}^{\vec{v}}}
\left[\sum_{i\in{\cal W}^{\vec{v}}}t_i^{\vec{v}}\cdot
\beta_{t_i^{\vec{v}}}(\vec{b}_{-i})
\right]\right]
\label{equation:left-hand-2}
\end{align}

\noindent Note that in the first term above, we sum only over $\vec{v}\in{\cal U}^i$, since for $\vec{v}\not\in{\cal U}^i$, $t_i^{\vec{v}}=0$.
By~\eqref{equation:sum1-9},~\eqref{equation:left-hand-2} and~\eqref{equation:bne-guarantee}, we obtain:

\begin{equation}
\mathbb{E}_{\vec{v}}
\left[\mathbb{E}_{\vec{b}\sim\vec{B}^{\vec{v}}}
\left[
\sum_iu_i(\vec{b})+
\sum_{i\in{\cal W}^{\vec{v}}}t_i^{\vec{v}}\cdot
\beta_{t_i^{\vec{v}}}(\vec{b}_{-i})
\right]\right]
\geq
\frac{1}{2}\mathbb{E}_{\vec{v}}[SW(\vec{x^{\vec{v}}})]
\label{equation:bne-bound-pre-final}
\end{equation}

\noindent To finish the proof, we substitute the second sum inside the expectation by its upper bound as given by Lemma~\ref{lemma:sw-cover} in~\eqref{equation:sw-cover}. Notice that $p(\vec{b})$ appearing in~\eqref{equation:sw-cover} is absorbed by the payment appearing in the utility $u_i(\vec{b})$ of at least one bidder. Thus we obtain:

\[
\mathbb{E}_{\vec{v}}
\left[\mathbb{E}_{\vec{b}\sim\vec{B}^{\vec{v}}}
\left[
\sum_iv_i\Bigl(x_i(\vec{b})\Bigr)+\frac{k-1}{k}SW(\vec{b})
\right]\right]
\geq
\frac{1}{2}\mathbb{E}_{\vec{v}}[SW(\vec{x^{\vec{v}}})]
\]
\noindent which essentially concludes the proof, by $\sum_iv_i\Bigl(x_i(\vec{b})\Bigr) = SW(\vec{b})$.\qed
\end{proof}

\bigskip

\noindent To complete our arguments for the proof of Theorem~\ref{theorem:bne-poa}, we prove Lemma~\ref{lemma:sw-cover}. Subsequently, we comment on how our arguments can be adjusted for the case of Bayes-Nash equilibria supported by general no-overbidding strategies (i.e., not restricted to undominated ones), to yield an upper bound of $4$ for the Price of Anarchy.

\medskip

\noindent {\bf Proof of Lemma~\ref{lemma:sw-cover}.}
For the winning bidders $i=1,\dots,\ell$ we have $x_i^*\geq 1$. Define $\psi_i=\sum_{j\leq i}x_j^*$. First, we prove by induction on $i$ that: 

\begin{equation}
\sum_{1\leq j\leq i}t_j\beta_{t_j}(\vec{b}_{-j})
\leq\beta_{t_1}(\vec{b}_{-1})+\sum_{1\leq j\leq \psi_i-1}\beta_j(\vec{b}).
\label{equation:sw-cover-induction}
\end{equation}

 \noindent For the basic step of the induction, consider $i=1$; then:
\begin{align}
\displaystyle
t_1\beta_{t_1}(\vec{b}_{-1})
&\leq\displaystyle
\beta_{t_1}(\vec{b}_{-1})+(t_1-1)\beta_{t_1}(\vec{b})
\label{equation:shift}\\
&\leq\displaystyle
\beta_{t_1}(\vec{b}_{-1})+\sum_{j=t_1}^{2t_1-2}\beta_j(\vec{b})\nonumber\\
&=\displaystyle
\beta_{t_1}(\vec{b}_{-1})+\sum_{j\leq \psi_1-1}\beta_j(\vec{b})\nonumber
\end{align}
Assuming~\eqref{equation:sw-cover-induction} holds for $i>1$, we show it remains true for $i+1$. We have:
\begin{align}
\sum_{j\leq i+1}t_j\beta_j(\vec{b}_{-j})
&=\displaystyle
\sum_{j\leq i}t_j\beta_j(\vec{b}_{-j})+
\Bigl(t_{i+1}\beta_{t_{i+1}}(\vec{b}_{-(i+1)})\Bigr)\nonumber\\
&\leq\displaystyle
\beta_{t_1}(\vec{b}_{-1})+\sum_{j\leq \psi_i-1}\beta_j(\vec{b})+
\Bigl(t_{i+1}\beta_{t_{i+1}}(\vec{b})\Bigr)\nonumber\\
&\leq\displaystyle
\beta_{t_1}(\vec{b}_{-1})+\sum_{j\leq \psi_{i+1}-1}\beta_j(\vec{b})
\label{equation:pause}
\end{align}
where~\eqref{equation:pause} is justified as follows: for any value of $\psi_i\geq 1$, we have $\beta_{t_{i+1}}(\vec{b})\leq\beta_j(\vec{b})$, for all
$j=\psi_i-1+\left\lceil\frac{x_{i+1}^*}{2}\right\rceil,\dots,\psi_i-1+x_{i+1}^*$ (a total of at least $\left\lceil\frac{x_{i+1}^*}{2}\right\rceil=t_{i+1}$ inequalities). By $\psi_{i+1}=\psi_i+x_{i+1}^*$, we obtain~\eqref{equation:pause} and, thus,~\eqref{equation:sw-cover-induction}. Setting $i=\ell$, thus, $\psi_{\ell}=k$ in~\eqref{equation:sw-cover-induction}, we have:
\begin{equation}
\displaystyle
\sum_{1\leq i\leq\ell}t_i\beta_{t_i}(\vec{b}_{-i})
\leq\displaystyle
\beta_{t_1}(\vec{b}_{-1})+\sum_{j=1}^{k-1}\beta_j(\vec{b})
\leq\displaystyle
\beta_{t_1}(\vec{b}_{-1})+SW(\vec{b})-\beta_k(\vec{b})
\label{equation:sw-pre-cover}
\end{equation}

For every bidder $i$ with $x_i^*\geq 1$, it must be also $x_i(\vec{b})\geq 1$, because the profile $\vec{b}$ consists of undominated strategies, thus, $b_i(1)=m_i(1)=v_i(1)$. Because $\beta_k(\vec{b})=\max_{i,j}b_i(j)$, by submodularity of valuation functions and by no-overbidding with respect to marginal bids, there exists a bidder $i_1$ such that $b_{i_1}(1)=v_{i_1}(1)=m_{i_1}(1)=\beta_k(\vec{b})$. Then, $\beta_k(\vec{b})=m_{i_1}(1)$ is also the largest marginal value contributed to $SW(\vec{b})$ and $m_i(1)\geq SW(\vec{b})/k$. Moreover, for every bidder $i$, it must hold that $b_i(1)=m_i(1)=v_i(1)\geq\beta_{x_i^*}(\vec{b})$, i.e., $b_i(1)$ is at least the $x_i^*$-th bid in $\beta_1(\vec{b}),\dots,\beta_k(\vec{b})$; otherwise, at least one bid in $\vec{b}$ is beyond the marginal value of some bidder, which also contradicts $\vec{b}$ consisting of undominated strategies. Then, for any bidder $i$ with $x_i^*\geq 1$ we have that: $\beta_{t_i}(\vec{b}_{-i})\leq p(\vec{b})$, because at least $x_i^*\geq t_i$ non-winning marginal bids in $\vec{b}$ (including $p(\vec{b})$) become the lowest winning in $\vec{b}_{-i}$. This yields $\beta_{t_1}(\vec{b}_{-1})\leq p(\vec{b})$ and we obtain~\eqref{equation:sw-cover} from~\eqref{equation:sw-pre-cover}.\qed

\medskip

We used our assumption of undominated support essentially only in the analysis that follows~\eqref{equation:sw-pre-cover}. If we replace this assumption with the standard no-overbidding assumption, we can continue our analysis from~\eqref{equation:sw-pre-cover}, by noticing that $\beta_{t_1}(\vec{b}_{-1})\leq\beta_{t_1}(\vec{b}_{-1})\leq\beta_k(\vec{b})$. Thus,~\eqref{equation:sw-cover} can be replaced by:

$$\sum_{1\leq i\leq\ell}t_i\beta_{t_i}(\vec{b}_{-i})\leq SW(\vec{b}).$$

\noindent We can use this latter upper bound in~\eqref{equation:bne-bound-pre-final}, along with $\sum_iu_i(\vec{b})\leq\sum_iv_i(\vec{b})=SW(\vec{b})$, to obtain an upper bound of $4$ on the Bayes-Nash Price of Anarchy.

\begin{corollary}
The mixed Bayes-Nash Price of Anarchy of the Uniform Price Auction for non-overbidding bidders with submodular valuation functions is at most $4$.
\label{corollary-2}
\end{corollary}

\newpage

\section*{Appendix A: Omitted Proofs from Section~\ref{section:undominated-pne}}

\subsection*{Proof of Lemma~\ref{lm:no-overbidding}}

Fix any bidder $i$ with a submodular valuation function $v_i$ and consider a bidding vector $\vec{b}_i=(b_i(1),\dots,b_i(k))$ of $i$, where $b_i(1)\geq b_i(2) \geq ... \geq b_i(k)$, as required by the bidding rule of the auction. 
Suppose that for some $j$, bidder $i$ overbids the marginal value for the $j$-th unit, i.e., $b_i(j) > m_i(j)$.
We will construct a bidding vector $\vec{b}_i'$ and we will show that $\vec{b}_i'$ weakly dominates $\vec{b}_i$. We define $\vec{b}_i'$ as follows: $b_i'(r) = b_i(r)$ for any $r\leq j-1$, and $b_i'(r) =m_i(r)$ for every $r\geq j$. Note that this is a valid bidding vector for the auction (this holds 
because $b_i'(j-1) = b_i(j-1) \geq b_i(j) > m_i(j)$).
We show that: for every $\vec{b}_{-i}$, $u_i(\vec{b}_i',\vec{b}_{-i})\geq u_i(\vec{b}_i,\vec{b}_{-i})$ and the inequality is strict for at least one vector $\vec{b}_{-i}$.

For a configuration $\vec{b}_{-i}$ of the other bidders, let $p(\vec{b})$ denote the uniform price under $\vec{b} = (\vec{b}_i, \vec{b}_{-i})$.
We start first, with configurations $\vec{b}_{-i}$ for which $j > x_i(\vec{b})$. In this case, if bidder $i$ plays according to $\vec{b}_i'$, he will retain at least the same utility, because $b_i(j)$ does not grant him any unit and neither does $m_i(j)$. Hence, he will keep winning under $\vec{b}_i'$ the same number of units and if his $j$-th bid was the price-setting bid, the price may even decrease and he will be strictly better. 
Consider now configurations $\vec{b}_{-i}$ for which $j\leq x_i(\vec{b})$. We examine the following subcases:

\medskip

\noindent{\bf (I):} $p(\vec{b}) < m_i(x_i(\vec{b}))$. Then, by playing $\vec{b}_i'$, bidder $i$ will still win the same number of units as before. The price under ($\vec{b}_i'$, $\vec{b}_{-i}$) will either remain the same or may decrease in the case that the price-setting bid was an overbid by bidder $i$.

\medskip

\noindent{\bf (II):} $p(\vec{b}) = m_i(x_i(\vec{b}))$. In this case, various scenarios may occur depending on the possible configurations for $\vec{b}_{-i}$, and on the possible appearance of ties. First, by switching to $\vec{b}_i'$, bidder $i$ will either win the same number of units as before, or he may win less if some tie is resolved against him. In the cases that he wins the same number of units as in $\vec{b}$, it is easy to see that the price has either remained the same or it may even have fallen, e.g., if the price-setting bid in $\vec{b}$ was $b_i(x_i(\vec{b})+1) > m_i(x_i(\vec{b})+1)$, which is reduced in $\vec{b}_i'$ to $m_i(x_i(\vec{b})+1)$. Hence bidder $i$ has at least the same utility as before. In the cases where bidder $i$ wins less units, the only possibility is that the price has remained the same and it is only because of ties that $i$ lost some of his previously won units. But then for the units that he lost, their marginal value is the same as the price, hence bidder $i$ simply had zero utility for them in $\vec{b}$. Thus he will still have the same utility under $\vec{b}_i'$ as before. 

\medskip

\noindent {\bf (III):} $p(\vec{b}) > m_i(x_i(\vec{b}))$. Then in $\vec{b}$, $i$ overpays his marginal value for at least one won unit; in this case bidder $i$ can switch to $\vec{b}_i'$ and {\em strictly increase} his utility. To see this, note that under ($\vec{b}_i', \vec{b}_{-i}$) bidder $i$ will still win the units that give him nonnegative utility (i.e., have marginal value at least $p(\vec{b})$), and he will lose only units that he could not afford anyway. 

\medskip

\noindent Conclusively, overbidding any marginal value is a weakly dominated strategy for every bidder $i$.\qed

\subsection*{Proof of Lemma~\ref{lm:vi(1)}}

By Lemma \ref{lm:no-overbidding}, undominated strategies are conservative w.r.t. marginal bids, thus no bidder may exaggerate his bid for the first unit. Let $\vec{b}_i$ denote such a bidding vector for bidder $i$, where $b_i(1)<m_i(1)$. Let $\vec{b}_i'$ denote the bidding vector where $b_i'(1)=m_i(1)$ and $b_i'(j)=b_i(j)$, for $j=2, \dots, k$. For any configuration $\vec{b}_{-i}$ due to all other bidders, we show that: $u_i(\vec{b}_i',\vec{b}_{-i})\geq u_i(\vec{b}_i',\vec{b}_{-i})$ and the inequality is strict for at least one such $\vec{b}_{-i}$.

If $x_i(\vec{b}_i,\vec{b}_{-i})\geq 1$, bidder $i$ will maintain his allocation and his utility by increasing $b_i(1)$ to $m_i(1)=b_i'(1)$. If $x_i(\vec{b}_i, \vec{b}_{-i})=0$, then an increase of $b_i(1)$ to $m_i(1)=b_i'(1)$ will either maintain the utility of $i$ to $0$ (if the current minimum winning bid is at least equal to $m_i(1)$), or increase it (if the current minimum winning bid is less than $m_i(1)$). In the latter case, $i$ wins one unit and the minimum winning bid of $(\vec{b}_i,\vec{b}_{-i})$ becomes the new uniform price that $i$ pays exactly once. Conclusively, bidding $b_i(1) < v_i(1)$ is a weakly dominated strategy.\qed

\newpage

\section*{Appendix B: Proof of Lemma~\ref{lemma:poa-general-ub}}

Let $\vec{b}$ be any pure Nash equilibrium configuration in undominated strategies. For simplicity we use $x_i$ for $x_i(\vec{b})$ and $\beta_j$ for $\beta_j(\vec{b})$, $j=1,\dots, k$. For $SW(\vec{b})$ we have:

\begin{equation}
SW(\vec{b}) = 
\sum_{i\in{\cal W}_0(\vec{x})}v_i(x_i)+
\sum_{i\in{\cal W}_1(\vec{x})}v_i(x_i)+
\sum_{i\in{\cal W}_2(\vec{x})}v_i(x_i)
\label{equation:sw-partition}
\end{equation}

\noindent By definition of ${\cal W}_0$, we can write the first term of~\eqref{equation:sw-partition}, $\sum_{i\in{\cal W}_0(\vec{x})}v_i(x_i)$, as:

\begin{align}
\sum_{i\in{\cal W}_0(\vec{x})}\Bigl(
v_i(x_i^*)+v_i(x_i)-v_i(x_i^*)
\Bigr)
&=
\sum_{i\in{\cal W}_0(\vec{x})}\left(
v_i(x_i^*)+
\sum_{j=1+x_i^*}^{x_i}m_i(j)
\right)\nonumber\\
&\geq
\sum_{i\in{\cal W}_0(\vec{x})}\left(
v_i(x_i^*)+\sum_{j=1+x_i^*}^{x_i}b_i(j)
\right)
\label{equation:first-term}
\end{align}

\noindent The last inequality is due to the fact that $\vec{b}$ is an {\em undominated} pure Nash equilibrium, thus, Lemma~\ref{lm:no-overbidding} applies. For the third term of~\eqref{equation:sw-partition}, we have similarly:

\begin{equation}
\sum_{i\in{\cal W}_2(\vec{x})}v_i(x_i)
=
\sum_{i\in{\cal W}_2(\vec{x})}\sum_{j=1}^{x_i}m_i(j)
\geq
\sum_{i\in{\cal W}_2(\vec{x})}\sum_{j=1}^{x_i}b_i(j)
\label{equation:third-term}
\end{equation}

\noindent Substituting~\eqref{equation:first-term} and~\eqref{equation:third-term} in~\eqref{equation:sw-partition}, we obtain:

\begin{equation*}
SW(\vec{b}) 
\geq 
\sum_{i\in{\cal W}_0(\vec{x})}\left(
v_i(x_i^*)+\sum_{j=1+x_i^*}^{x_i}b_i(j)
\right)+
\sum_{i\in{\cal W}_1(\vec{x})}v_i(x_i)+
\sum_{i\in{\cal W}_2(\vec{x})}\sum_{j=1}^{x_i}b_i(j)
\end{equation*}

\noindent Observe that, for every unit missed under $\vec{b}$ by any bidder
$i\in{\cal W}(\vec{x}^*)\cap{\cal W}_1(\vec{x})$, there exists a bidder $i'\in{\cal W}_0(\vec{x})\cup {\cal W}_2(\vec{x})$ that obtains this unit. If $i$ missed $x_i^*-x_i>0$ units in $\vec{b}$, there are at least as many bids issued by bidders in ${\cal W}_0(\vec{x})\cup {\cal W}_2(\vec{x})$, that won collectively these units. These bids sum up to at least $\sum_{j=1}^{x_i^*-x_i}\beta_j$, i.e. the sum of the $x_i^*-x_i$ lowest winning bids in $\vec{b}$. I.e.:
\begin{equation}
\sum_{i\in{\cal W}_0(\vec{x})}\sum_{j=1+x_i^*}^{x_i}b_i(j)+
\sum_{i\in{\cal W}_2(\vec{x})}\sum_{j=1}^{x_i}b_i(j)
\geq
\sum_{i\in{\cal W}_1(\vec{x})}\sum_{j=1}^{x_i^*-x_i}\beta_j
\label{equation:first-third-term}
\end{equation}
\noindent Using~\eqref{equation:first-third-term} in the last lower bounding expression of $SW(\vec{b})$, we obtain:
%
\begin{equation}
SW(\vec{b})\geq
\sum_{i\in {\cal W}_0(\vec{x})}v_i(x_i^*)+
\sum_{i\in {\cal W}_1(\vec{x})}\Bigl(v_i(x_i)+\sum_{j=1}^{x_i^*-x_i}\beta_j\Bigr)
\label{equation:conservative-pne}
\end{equation}

\noindent Finally, we notice that ${\cal W}_0(\vec{x})$ and ${\cal W}_1(\vec{x})$ are a partition of ${\cal W}(\vec{x}^*)$, thus:

\begin{equation}
SW(\vec{b}^*) = \sum_{i\in{\cal W}_0(\vec{x})}v_i(x_i^*)+\sum_{i\in{\cal W}_1(\vec{x})}v_i(x_i^*)
\label{equation:opt-partition}
\end{equation}

\noindent By~\eqref{equation:conservative-pne}
and~\eqref{equation:opt-partition} we obtain~\eqref{equation:poa-general-ub}.
\qed
\newpage

\section*{Appendix C: Omitted Proofs from Section~\ref{section:bne-poa}}

\subsection*{Proof of Lemma~\ref{lemma:bayesian-lemma}}

Fix any player $i$. The statement trivially holds for $j=0$, hence consider $j\geq 1$.
For any bidding configuration $\vec{b}_{-i}$, we let
$\beta_1(\vec{b}_{-i})\leq\cdots\leq\beta_k(\vec{b}_{-i})$ denote the winning
bids, if $i$ was not present, in non-decreasing order. 
Fix also the index $j$ ($1\leq j \leq k$), and using the submodular valuation function $v_i$ of $i$, define the bidding vector $\vec{m}_i^{[j]}=(m_i(1),\dots,m_i(j),0,\dots,0)$. 

Assume that bidder $i$ has won $s$ units in the configuration $(\vec{m}_i^{[j]}, \vec{b}_{-i})$. Obviously $s\leq j$. Let $p$ denote the price that he pays for these units. In the case that $s=j$, $p = \beta_j(\vec{b}_{-i})$, and the statement of the Lemma trivially holds. In the case that $s<j$, we have that $p = \max \{m_i(s+1), \beta_s(\vec{b}_{-i}) \}$. This is because the highest losing bid may be either the next bid of bidder $i$ (i.e., $m_i(s+1)$), or the highest losing bid in $\vec{b}_{-i}$, which is $\beta_s(\vec{b_{-i}})$. This implies that $p \leq \beta_j(\vec{b}_{-i})$. Thus we can derive the following bound on the utility of bidder $i$:

\[
\begin{array}{lcl}
u_i(\vec{m}_i^{[j]},\vec{b}_{-i})
& \geq &
v_i(s)-s\cdot\beta_{j}(\vec{b}_{-i})
\smallskip\\
& \geq & \displaystyle
v_i(s) - s\cdot\beta_{j}(\vec{b}_{-i}) + \sum_{r=1}^{j-s} (m_i(s+r) - \beta_{j}(\vec{b}_{-i}))\smallskip\\
& = & \displaystyle
v_i(s) + \sum_{r=1}^{j-s} m_i(s+r) -j\cdot\beta_{j}(\vec{b}_{-i})\smallskip\\
& = & v_i(j)-j\cdot\beta_{j}(\vec{b}_{-i})\smallskip\\

\end{array}
\]

\noindent The second inequality above holds because $m_i(s+r)\leq m_i(s+1) \leq \beta_{j}(\vec{b}_{-i})$ for all marginal values beyond the 
$s$-th one. This completes the proof.\qed

\subsection*{Omitted Part of Proof of Theorem~\ref{theorem:bne-poa}}

Consider the first of the two summands of~(\ref{equation:bne-guarantee}). We
explain the derivations below.

\begin{align}
&\displaystyle
\sum_i
\sum_{\vec{v}\in{\cal V}}
\pi(\vec{v})
\mathbb{E}_{\vec{w}_{-i}|v_i}\Biggl[
\mathbb{E}_{\vec{b}\sim\vec{B}^{(v_i,\vec{w}_{-i})}}
\Bigl[u_i(\vec{b})\Bigr]\Biggr]
\label{equation:sum1-0}\\
&=\displaystyle
\sum_i\sum_{\vec{v}_{-i}\in{\cal V}_{-i}}\pi(\vec{v}_{-i}|v_i)
\sum_{v_i\in V_i}\pi_i(v_i)
\sum_{\vec{w}_{-i}\in{\cal V}_{-i}}\pi(\vec{w}_{-i}|v_i)
\mathbb{E}_{\scriptscriptstyle{\vec{b}\sim\vec{B}^{(v_i,\vec{w}_{-i})}}}
\Bigl[u_i(\vec{b})\Bigr]
\nonumber\\
&=\displaystyle
\sum_i\sum_{v_i\in V_i}\pi_i(v_i)
\sum_{\vec{v}_{-i}\in{\cal V}_{-i}}\pi(\vec{v}_{-i}|v_i)
\sum_{\vec{w}_{-i}\in{\cal V}_{-i}}\pi(\vec{w}_{-i}|v_i)
\mathbb{E}_{\scriptscriptstyle{\vec{b}\sim\vec{B}^{(v_i,\vec{w}_{-i})}}}
\Bigl[u_i(\vec{b})\Bigr]
\nonumber\\
&=\displaystyle
\sum_i\sum_{v_i\in V_i}\pi_i(v_i)
\sum_{\vec{w}_{-i}\in{\cal V}_{-i}}\pi(\vec{w}_{-i}|v_i)
\mathbb{E}_{\vec{b}\sim\vec{B}^{(v_i,\vec{w}_{-i})}}
\Bigl[u_i(\vec{b})\Bigr]
\label{equation:sum1-7}\\
&=\displaystyle
\sum_{\vec{v}\in {\cal V}}\pi(\vec{v})\sum_i
\mathbb{E}_{\vec{b}\sim\vec{B}^{(v_i,\vec{w}_{-i})}}
\Bigl[u_i(\vec{b})\Bigr]
\label{equation:sum1-8}\\
&=\displaystyle
\sum_{\vec{v}\in{\cal V}}\pi(\vec{v})
\mathbb{E}_{\vec{b}\sim\vec{B}^{\vec{v}}}
\left[\sum_iu_i(\vec{b})\right]=
\mathbb{E}_{\vec{v}}
\Biggl[\mathbb{E}_{\vec{b}\sim\vec{B}^{\vec{v}}}
\left[\sum_iu_i(\vec{b})\right]\Biggr]
\label{equation:sum1-9}
\end{align}
To obtain the equality following~\eqref{equation:sum1-0}, we analyze $\sum_{\vec{v}\in{\cal V}}\pi(\vec{v})$ to: 
$$
\sum_{\vec{v}_{-i}\in{\cal V}_{-i}}\pi(\vec{v}_{-i})\sum_{v_i\in V_i}\pi_i(v_i)=\sum_{(v_i,\vec{v}_{-i})}\pi(\vec{v}_{-i}|v_i)\pi(v_i)
$$
For~\eqref{equation:sum1-7}, it suffices to move $\sum_{\vec{v}_{-i}\in{\cal V}_{-i}}\pi(\vec{v}_{-i}|v_i)$ all the way to the right and observe that it equals $1$ by independence of the valuation distributions. Finally, we obtain~(\eqref{equation:sum1-8} by condensing:
$$
\displaystyle
\sum_i
\sum_{v_i\in V_i}\pi_i(v_i)
\sum_{\vec{w}_{-i}\in{\cal V}_{-i}}\pi(\vec{w}_{-i}|v_i)
$$
into $\sum_{\vec{v}\in{\cal V}}\pi(\vec{v})$ and moving $\sum_i$ to the right, to sum over the expected valuations of bidders for every valuation profile $\vec{v}\in{\cal V}$.

\end{document}